\documentclass[sigconf]{acmart}

\AtBeginDocument{%
  \providecommand\BibTeX{{%
    \normalfont B\kern-0.5em{\scshape i\kern-0.25em b}\kern-0.8em\TeX}}}

\copyrightyear{2022} 
\acmYear{2022} 
\setcopyright{acmlicensed}\acmConference[SPAA '22]{Proceedings of the 34th ACM Symposium on Parallelism in Algorithms and Architectures}{July 11--14, 2022}{Philadelphia, PA, USA}
\acmBooktitle{Proceedings of the 34th ACM Symposium on Parallelism in Algorithms and Architectures (SPAA '22), July 11--14, 2022, Philadelphia, PA, USA}
\acmPrice{15.00}
\acmDOI{10.1145/3490148.3538582}
\acmISBN{978-1-4503-9146-7/22/07}

\usepackage{outlines, array}
\newtheorem{thm}{Theorem}
\newtheorem{claim}[thm]{Claim}

\begin{document}

\title{Balancing Flow Time and Energy Consumption}

\author{Sami Davies}
\authornote{All authors contributed equally to this research.}
\authornote{Supported by a NSF CI Fellowship.}
\email{sami@northwestern.edu}
\affiliation{Computer Science Department
\institution{Northwestern University}
\streetaddress{2233 Tech Drive}
\city{Evanston, IL}
\country{USA}}

\author{Samir Khuller}
\authornotemark[1]
\authornote{Supported by an Adobe Research award and an Amazon Research award.}
\email{samir.khuller@northwestern.edu}
\affiliation{Computer Science Department
\institution{Northwestern University}
\streetaddress{2233 Tech Drive}
\city{Evanston, IL}
\country{USA}}

\author{Shirley Zhang}
\authornotemark[1]
\email{skzhang@alumni.princeton.edu}
\affiliation{Computer Science Department
\institution{Northwestern University}
\streetaddress{2233 Tech Drive}
\city{Evanston, IL}
\country{USA}}


\begin{abstract}

In this paper, we study the following batch scheduling model: find a schedule that minimizes total flow time for $n$ uniform length jobs, with release times and deadlines, where the machine is only actively processing jobs in at most $k$ synchronized batches of size at most $B$. Prior work on such batch scheduling models has considered only feasibility with no regard to the flow time of the schedule. However, algorithms that minimize the cost from the scheduler's perspective---such as ones that minimize the active time of the processor---can result in schedules where the total flow time is arbitrarily high \cite{ChangGabowKhuller}. Such schedules are not valuable from the perspective of the client. In response, our work provides dynamic programs which minimize flow time subject to active time constraints. Our main contribution focuses on jobs with agreeable deadlines; for such job instances, we introduce dynamic programs that achieve runtimes of O$(B \cdot k \cdot n)$ for unit jobs and O$(B \cdot k \cdot n^5)$ for uniform length jobs. These results improve upon our modification of a different, classical dynamic programming approach by Baptiste. While the modified DP works when deadlines are non-agreeable, this solution is more expensive, with runtime $O(B \cdot k^2 \cdot n^7)$ \cite{Baptiste00}.

\end{abstract}




\maketitle

\section{Introduction}

There has been an increasing focus in the scheduling literature on power conservation and minimizing the energy consumption of the processor \cite{BaptisteChrobakDurr07, ChangGabowKhuller, ChangKhullerMukherjee, FinemanSheridan15, KSST-busytime, KoehlerKhuller17-busytime, KumarKhuller18, SimonFalkMegowTeich20, AntoniadisGargKumarKumar20}. One motivation for developing models that consider processor time is the desire to minimize environmental and financial costs at data centers where accessing memory is expensive. There are several natural ways in which energy consumption can be taken into account. If the cost of turning the processor on is prohibitive, reasonable objective functions include minimizing the gaps between scheduled jobs or maximizing the intervals with consecutively scheduled jobs \cite{Baptiste06, BaptisteChrobakDurr07, FinemanSheridan15, AntoniadisGargKumarKumar20}. On the other hand, if a machine has a high cost whenever it is on but a relatively small setup cost, it is reasonable to instead minimize the amount of time the machine is on \cite{ChangGabowKhuller, KumarKhuller18, ChangKhullerMukherjee, KSST-busytime, KoehlerKhuller17-busytime}. We consider active time constraints of the second form.

Introduced by Chang, Gabow, and Khuller,  
the \emph{active time problem} considers a set of $n$ unit jobs, each with a release time and deadline, where the goal is to schedule the jobs on a single machine in a minimal number of batches, such that each batch contains at most $B$ jobs \cite{ChangGabowKhuller}.
Here, a batch is a group of jobs that can be performed on the machine together and $B$ is the maximum capacity of the processor. Note that minimizing the number of batches is equivalent to minimizing the number of active slots when jobs have uniform length.  When jobs are non-unit length, and can be scheduled preemptively the active time problem is known to be \textsf{NP}-complete, and there exists several 2-approximation algorithms \cite{SahaPurohit21, ChangKhullerMukherjee, KumarKhuller18, CalinescuWang21}.
All prior research on batch scheduling before the work of Chang, Gabow, and Khuller focused on finding feasible schedules without worrying about minimizing the number of batches. However, if each active time slot is expensive, it is natural to want to minimize cost by minimizing the number of batches. The active time model cleanly captures the difficulty in applications such as minimizing the fiber costs of Optical Add Drop Multiplexers (OADMs) and VM consolidation in cloud computing \cite{FlamminiMMSSTZ10, ChauL20}.  

Let $[n]$ denote the set of jobs to schedule, and let $j \in [n]$ denote a single job with integral release time $r_j$ and deadline $d_j$.  We consider the setting where all jobs have length $p \in \mathbb{N}$, and each job $j$ must be scheduled at $p$ consecutive time slots, i.e. non-preemptively, in the interval $[r_j,d_j]$. 
The processor performs synchronous batching, so a batch cannot start until the last is finished, even if it was not full.
Some of our results are in the unit length setting, where $p=1$. We also focus on results with the practical, well-studied assumption that for all $i,j \in [n]$, if $r_i < r_j$ then $d_i \leq d_j$; such deadlines are called agreeable \cite{BampisDKM12, AngelBC14, AlbersMS14, KononovK20}.
The completion time of job $j$ is denoted by $C_j$, so $C_j \in [r_j+p,d_j]$. Up to $B$ jobs can be scheduled in a batch, and the flow time of a schedule with completion times $\{C_j\}$ is $\sum_{j} (C_j-r_j)$. Note that for uniform length jobs, minimizing flow time and wait time are equivalent as flow time is equal to wait time plus the processing length of the job. Due to integrality assumptions, it is without loss of generality to assume that time is slotted.
Given a set of $n$ jobs of length $p$---each equipped with their integral release times and deadlines---and a budget of $ p \cdot k$ active time slots, our goal is to find an assignment of the jobs to time slots that minimizes flow time such that at most $ p \cdot k$ time slots (or equivalently, $k$ batches) are active.

There is an inherent trade-off between the flow time of a schedule and the number of active time slots of the machine, which we will exemplify with a shuttle bus service. 
Consider a shuttle service that sends shuttle buses out to take riders from a parking garage to an event.
The shuttle service wants to minimize costs---or CO$_2$ emissions---by sending as few shuttles as possible, but at the same time the service wants to keep their passengers happy by minimizing passengers' average wait time. These two objectives are not symbiotic. If we only wanted to minimize wait time, we would send a shuttle per passenger, and if we tried to batch passengers as much as possible, the earlier passengers on a given shuttle may have to wait for a long time. Deadlines can ensure that no single customer is waiting for a very long time, and then
we can fix the maximum number of shuttles we are willing to send out and minimize total wait time subject to the number of shuttles available.  If we had a fast algorithm to do this, we could then use that algorithm with a range of values for the hard constraint to find an active time versus flow time trade-off that is acceptable. See Figure \ref{fig: new-slot} for example schedules using $k' \leq k=3$ non-empty active batches.

\begin{figure}
    \centering
    \includegraphics[width = 8cm]{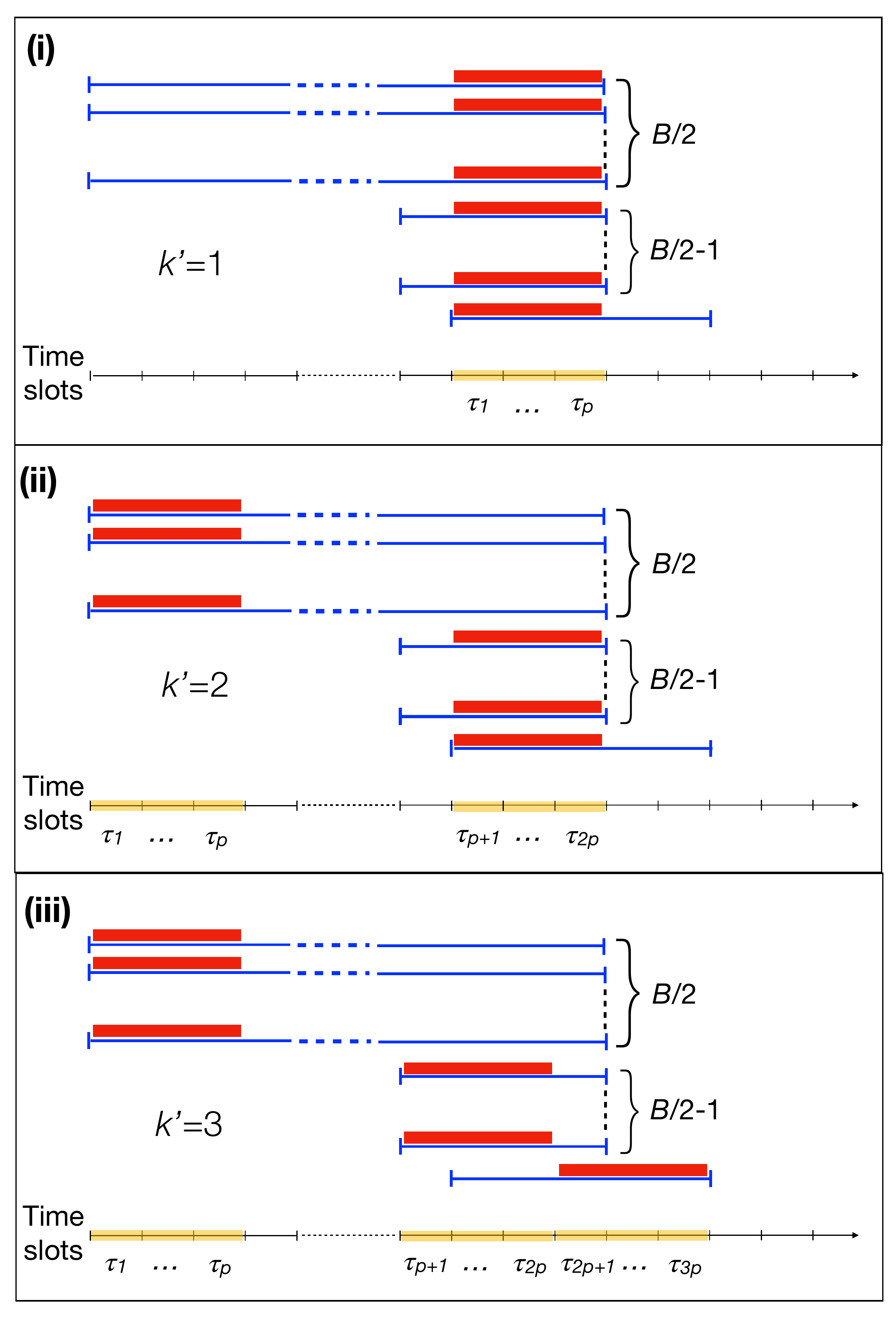}
    \caption{Fix $k=3$ to be the number of active batches. Throughout, $p$ is the length of the jobs, $k' \leq k =3$ is the number of non-empty active batches, and $B$ is the batch capacity. Blue intervals indicate job release time and deadline intervals. Red rectangles indicate where a job is scheduled, and highlighted time slots are active. 
The flow time drastically decreases from frame (i) to frame (ii), and either increases or barely decreases from frame (ii) to frame (iii).}
    \label{fig: new-slot}
\end{figure}

Overall, our work is applicable in the energy minimization setting when the cost of turning on a machine is low, but the cost of keeping it running is high. In such a setting, existing algorithms minimizing the total time the machine is on both for the active time problem and the busy time problem produce solutions with arbitrarily bad flow time. Agreeable deadlines and uniform length jobs are appropriate assumptions in this setting, as these assumptions are applicable to our real-world inspired motivations for this problem. In particular, uniform jobs are the central topic in batch scheduling models, which will be evident in the related work section, and agreeable deadlines can reasonably model customer wait time.

\subsection{Related Work}

In previous work, Baptiste provides a fundamental framework for using dynamic programming to study batch scheduling problems. 
In the serial batching model, the cost of a batch is the sum of the job lengths in the batch, and in the parallel batching model, the cost of the batch is the length of the longest job in the batch, where batches can have size at most $B$. Baptiste showed that in both models a large collection of objective functions are solvable in polynomial time with dynamic programming. For jobs of uniform length, Baptiste's dynamic program runs in time $O(B \cdot k^2 \cdot n^7)$ \cite{Baptiste00}. 
It is important to note that many of the works cited in the rest of this related work discussion are heavily influenced by this initial dynamic programming approach from Baptiste (see, for instance, \cite{Baptiste06, BaptisteChrobakDurr07, ChangGabowKhuller}).

The active time problem was introduced by Chang, Gabow, and Khuller \cite{ChangGabowKhuller}. They presented a greedy algorithm called Lazy Activation in the case when jobs have unit length. Lazy Activation minimizes the number of batches by delaying jobs; however, this potentially leads to a solution with very high flow time. We discuss this algorithm in more detail at the end of this section.
When jobs are pre-emptive but do not have unit processing time, Chang, Khuller, and Mukherjee gave an involved LP rounding based 2-approximation \cite{ChangKhullerMukherjee}. Subsequently, Kumar and Khuller gave a direct combinatorial 2-approximation algorithm \cite{KumarKhuller18}. It is worth noting that their algorithm also leads to solutions with arbitrarily bad flow time. 
Very recently, Calinescu and Wang introduce a new LP, for which they conjecture the integrality gap is 5/3. However, their rounding schemes can also only provably obtain a 2-approximation \cite{CalinescuWang21}. The active time problem has very recently been shown to be \textsf{NP}-complete by Saha and Purohit \cite{SahaPurohit21}.

When jobs have unit lengths, the active time problem can be viewed as a special case of the rectangle stabbing problem in computational geometry with capacity constrained lines. Even et. al. study this problem and use a DP method inspired by Baptiste that solves the problem in time $O(|S|^2|R|^2(|S|+R|))$, for $|R|$ the number of rectangles and $|S|$ the number of lines \cite{EvenLRSSS08}.  Other versions of this problem, with different objective functions, have also received attention \cite{ChanD0SW18, EGSV}.

When jobs have unit length, 
the problem of minimizing/ maximizing an objective subject to a budget of active time slots was studied under the name of the \emph{tall/small jobs problem}. In this setting, jobs are either considered "tall" and use all of the capacity available on a machine, or "small" and use 1 unit of available capacity.

Baptiste and Schieber study the unit length tall/small jobs problem with the objective of minimizing the maximum tardiness. They construct an LP for the problem whose solution must be integral, even though the LP is not totally unimodular \cite{BaptisteSchieber03}. On the other hand, D{\"{u}}rr and Hurand project the basic LP to a more compact LP of difference equations, which is shown to be totally unimodular \cite{DurrHurand06}. In general, the techniques in these papers cannot be adapted to handle the flow time objective, as the goal of minimizing flow time cannot be written as a feasibility problem.
Through the connection to the tall/ small jobs problem, we see that our techniques resemble those of Grandoni, Momke, and Wiese on the unsplittable flow on a path (UFP) problem  \cite{GMW21}.

The active time problem is very similar to the busy time problem, where the objective is to minimize the total amount of time that machines are on. In the literature, busy time is the non-preemptive version of the active time problem on an unbounded number of machines, as opposed to one machine in the active time problem.  
Approximation algorithms with small constant factors were given by Flammini et. al., Khandekar et. al., and Chang, Khuller, and Mukherjee with the best bound being a factor of 3 for the general case, and a factor of 2 for the case of interval jobs \cite{FMMS-busytime,KSST-busytime, ChangKhullerMukherjee}. 
Liu and Tang found constant factor approximation algorithms for the busy time problem on heterogeneous machines, where machines have different costs and different capacities \cite{LT21}.
Busy time was also studied in the online setting, both when machines are homogeneous and heterogeneous \cite{KoehlerKhuller17-busytime, LT21}. None of these works consider flow time as an objective -- note that minimizing flow time only makes sense in the general setting and not in the case of interval jobs.

Outside of the batch scheduling model, there are energy saving models that consider the activation cost of waking up a machine, when turning on the machine from a sleep state is very expensive. This is sometimes referred to as the power-down mechanism. 
Baptiste and then Baptiste, Chrobak, and D{\"u}rr study this model and show that on a single machine this problem is in $\textsf{P}$ when preemption and migration are allowed \cite{Baptiste06, BaptisteChrobakDurr07}. 
On multiple machines, again with preemption and migration, Demaine et. al. showed the Baptiste's dynamic program can be useful for obtaining a polynomial time algorithm on unit length jobs \cite{Baptiste06, DemaineGHSZ07}. 
Antoniadis, Garg, Kumar, and Kumar give the first constant factor approximation for scheduling arbitrary length jobs on $m$ machines \cite{AntoniadisGargKumarKumar20}.

Other energy minimization work is similar in spirit to ours, though the models are rather different. 
Baptiste showed that the problem of scheduling unit length jobs while minimizing the number idle periods is in $\textsf{P}$\cite{Baptiste06}.
Several works study minimizing energy consumption when consumption follows some power law \cite{AupyBDR11, SimonFalkMegowTeich20} or minimizing makespan given some energy budget \cite{PruhsSU05,BampisLL14}, and some of these works consider precedence-constrained models.
The Integrated Stockpile Evaluation problem seeks to minimize the number of calibrations, where a machine can only be used for a fixed period of length $T$ after it has been calibrated \cite{BenderBLMP13, FinemanSheridan15}.

The objective of minimizing flow time is one of the most important both in practice and theory. It is one of the most popular optimality criteria for scheduling on distributed batch systems and massively parallel processors \cite{Drozdowski09}. There has also been a surge of interest in studying weighted flow time on a single machine, where preemption is permitted. Batra, Garg, and Kumar made a breakthough on the problem in finding a pseudo-polynomial $O(1)$-approximation, which was then made polynomial by Feige, Kulkarni, and Li \cite{BatraGK18, FeigeKL19}. Even more recently, Rohwedder and Weise introduce a new method of attacking the problem, which leads to a $(2+\epsilon)$-approximation \cite{RohwedderW21}.

\subsection*{Lazy Activation}

Lazy Activation is an algorithm introduced by Chang, Gabow, and Khuller that minimizes the number of active slots needed to schedule a set of unit length jobs on single processor that can handle up to $B$ jobs at a time \cite{ChangGabowKhuller}. We briefly cover their algorithm here to motivate the need to develop new algorithms that take into account both flow time and active time. Lazy Activation first pre-processes all jobs so there are at most $B$ jobs with the same deadline. It then iterates: while there are unscheduled jobs remaining, Lazy Activation picks the unscheduled job $j$ with the earliest deadline and places an active slot $\tau$ as late as possible such that job $j$ can still be scheduled. It then schedules as many jobs as possible at $\tau$, choosing jobs using the earliest deadline first rule. 

Lazy Activation places active slots only when some unscheduled job is about to become infeasible. Simple examples show that Lazy Activation can be arbitrarily bad for minimizing flow time, e.g. consider one active slot for the job set containing $B$ unit jobs whose $(r_j, d_j) = (0, d)$ for some large $d$. In this case, the schedule returned by Lazy Activation has total flow time $B \cdot d$, while the optimal schedule has total flow time $B$. Lazy Activation can be run back-to-front by flipping release times and deadlines and treating the latest deadline as time 0, but this still results in solutions with arbitrarily bad flow time. For example, consider two active slots for $B-1$ jobs with $(r_j, d_j) = (0, d)$ and $1$ job with $(r_j, d_j) = (d - 1, d)$. In this case, Lazy Activation would schedule all jobs at time $d - 1$ and have total flow time $d(B - 1) + 1$ while the optimal schedule would schedule all but one job at time $1$ and have total flow time $B$.

\subsection{Our Contributions}

Our results are the first of to study the delicate balance between an active time budget and flow time, but they fit into the landscape of problems that consider energy efficiency without sacrificing an objective function. We believe the Pareto frontier between a budget of active time slots and an objective is understudied, and hope that this work will serve as inspiration for similar directions.

We introduce some notation and terminology to state our results. Assume jobs are ordered first by non-decreasing order of their release times and then by non-decreasing order of deadlines. Recall jobs have length $p \in \mathbb{Z}^+$. We refer to an instance of $\alpha \cdot p$, for $\alpha \in \{0,1,\ldots, k\}$, available active time slots with capacity $B$ and jobs $[j]$, for $j \in [n]$, equipped with their release times and deadlines with the shorthand $(\alpha,j)$. To keep notation more readable, we omit release times and deadlines from the shorthand, as well as the batch size. We may refer to $\alpha$ as the number of active batches. 
Our choice to refer to problems as $(\alpha,j)$ instead of $(k,n)$ in lemmas is to remind the reader that the properties we prove are true for all sub-problems of $(k,n)$, which is necessary since  all of our results use dynamic programming. Note that it is WLOG to assume that the number of active slots available is a multiple of $p$, i.e. the number of actiuve batches is integral.

At most $B$ jobs can be scheduled in an active time slot, where $B$ may be referred to as the \emph{batch size} or \emph{capacity}. Jobs cannot be scheduled in time slots that are not active.
We say that deadlines are \emph{agreeable} when for all jobs $i,j$, $d_i \leq d_j$ exactly when $r_i < r_j$. 
Recall that in the shuttle example, this would imply that if one rider arrives before another, the later rider would not be required to be picked up before the earlier one.  The additional assumption of agreeable deadlines is both practical and common in related literature. Theoretically it imposes additional structure on the problem which allows for a smaller search space for the DP to find an optimal solution. 

Our first result is for minimizing the flow time of $n$ unit length jobs with agreeable deadlines in $k$ active batches with capacity $B$.

\begin{thm}\label{main}
Let $(k,n)$ be an instance of unit length jobs with agreeable deadlines. Then one can either certify $(k,n)$ is infeasible or find a schedule for $(k,n)$ that minimizes flow time with a dynamic program in time $O(B \cdot k \cdot n)$ and space $O(k \cdot n)$.
\end{thm}

Additionally, we can add a parameter allowing us to only complete a subset of size $m  \leq n$ jobs. We consider $m$ as part of the input, and we wish to choose the $m$ jobs that will be completed such that flow time is minimized out of all $\binom{n}{m}$ possible choices. 
One instance where this generalization is interesting is when it is impossible to schedule all jobs in $[n]$ in $k$ active batches, so one considers a large subset instead. We present the extension of Theorem \ref{main} to completing a subset of jobs in Section \ref{sec: thm1}, and the runtime and space of the corresponding DP increase by a factor of  $n-m$.

Generalizing to non-unit jobs, we assume all jobs have length $p$. Here, there is less structure to use for the problem, which forces us to consider more possibilities for where to schedule batches. 
We obtain the following.

\begin{thm}\label{uniform}
Let $(k,n)$ be an instance of uniform length jobs with agreeable deadlines. Then one can either certify $(k,n)$ is infeasible or find a schedule for $(k,n)$ that minimizes flow time with a dynamic program in time $O(B \cdot k \cdot n^5)$ and space $O(k \cdot n^3)$.
\end{thm}

After seeing how to handle subsets in the unit length case, it is not hard to extend the uniform length setting to also handle subsets, and we present this extension in Section \ref{sec: uniform}.

Lastly, we show that our problem with arbitrary, i.e. potentially non-agreeable, deadlines can be solved by augmenting a dynamic program of Baptiste to include an active time budget. While this result is more general, its runtime is much worse. However, it is not hard to see that this can be slightly improved for unit jobs.

\begin{thm}\label{thm:Baptiste-ext}
Let $(k,n)$ be an instance of uniform length jobs. Then one can either certify $(k,n)$ is infeasible or find a schedule for $(k,n)$ that minimizes flow time with a dynamic program in time $O(B \cdot k^2 \cdot n^7)$ and space $O(B \cdot k \cdot n^5)$.
If $p=1$, then the runtime can be reduced to $O(B \cdot k^2 \cdot n^4)$ and the space complexity to $O(B \cdot k \cdot n^3)$.
\end{thm}

With arbitrary, i.e. potentially non-agreeable, deadlines, there may be computational barriers to obtaining algorithms with runtime comparable to ours in the agreeable deadlines setting. We discuss these difficulties more in Section \ref{sec: conclusion}. 

Our results are summarized in the table. As expected, the runtime and space complexity are best when there is more structure in the setting to use, i.e. unit length jobs or agreeable deadlines. The most structured settings require more technical care in order to fully utilize this additional structure, while the less structured settings need to use more brute force.

We also discuss the contributions of our techniques.
First, the fact that dynamic programming is useful for our problem was non-obvious, a priori. In the active time literature, combinatorial, LP rounding, and dynamic programming algorithms have all been useful to various extents. Second, it is important to note that our dynamic programs for agreeable jobs are \emph{not} generalizations or specializations of the dynamic program in Baptiste’s highly influential batch scheduling work. Baptiste's DP looks over all time intervals $(t_l, t_r)$, keeps track of how much each slot is ‘filled’ so far, and must schedule jobs one at a time. For details on Baptiste's framework, see Section \ref{Baptiste-sec}. Our DP, on the other hand, keeps track of the last time slot that a job is scheduled in and schedules jobs in consecutive batches, filling an entire slot at once.  Our hope is that our DP framework can be extended to more general settings, thus beating the runtime of Baptiste’s DP in more general settings. We only augment Baptiste’s DP to include an active time budget in the non-agreeable jobs setting, and this is more so to exemplify that our problem is in $\textsf{P}$ for non-agreeable jobs and to provide a baseline runtime and space complexity.

\vspace{.5cm}
\begin{center}
\begin{tabular}{ |m{1.8cm}||m{2.6cm}|m{2.7cm}| } 
 \hline
  job lengths $\rightarrow$ 
  deadlines  $\downarrow$   &   {\color{white} fillllll ll} unit &  {\color{white} filll llll} uniform  \\ 
 \hline
 \hline
    {\color{white} fil}agreeable & \vspace{2mm} Theorem \ref {main}, Section~\ref{sec: thm1}   runtime: $O(B \cdot k \cdot n)$  space: $O(k \cdot n)$   {\color{white} filllll}
  + subset ext. (Thm \ref{subset})  
   &  \vspace{2mm} Theorem \ref {uniform}, Section \ref{sec: uniform}
   runtime: $O(B \cdot k \cdot n^5)$  space: $O( k \cdot n^3)$ {\color{white} filll l}
     + subset ext. (Thm \ref{uniform-subset})  \\
  \hline
  {\color{white} fil} arbitrary &  \vspace{2mm} Theorem \ref {thm:Baptiste-ext}, Section \ref{Baptiste-sec} runtime: $O(B \cdot k^2 \cdot n^4)$ \qquad space: $O( B \cdot k \cdot n^3)$ 
   &  \vspace{2mm} Theorem \ref {thm:Baptiste-ext}, Section \ref{Baptiste-sec}   runtime: $O(B \cdot k^2 \cdot n^7)$ \qquad space: $O(B\cdot  k \cdot n^5)$  \\ 
 \hline
\end{tabular}
\end{center}
\vspace{.5 cm}

\section{Unit jobs, agreeable deadlines}\label{sec: thm1}
Here, we work towards proving Theorem \ref{main}, and presenting a dynamic program for the case when jobs have unit lengths and agreeable deadlines.
We begin by formalizing more definitions. 
An instance $(\alpha,j)$ is \emph{feasible} when it is possible to schedule all jobs in $[j]$ in $\alpha$ active time slots, respecting release times and deadlines. Note here since jobs are unit, the number of active time slots is the same as the number of active batches.
When $(\alpha,j)$ has unit length jobs and $B \cdot \alpha < j$, it is clearly an infeasible instance, as only $\alpha \cdot B$ jobs can be scheduled in $\alpha$ time slots with batch size $B$. It is also possible to have infeasible instances with $B \cdot \alpha \geq j$, but we will handle these as base cases in the DP formulation.
When $(\alpha,j)$ is feasible, then
$\textsf{OPT}(\alpha,j)$ is finite and denotes the minimum flow time of any optimal schedule that uses at most $\alpha$ time slots to schedule job set $[j]$. Additionally, for feasible $(\alpha,j)$, take  $\mathcal{S}(\alpha,j)$ to be the set of optimal schedules achieving flow time $\textsf{OPT}(\alpha,j)$. When $(\alpha,j)$ is infeasible, we let
$\textsf{OPT}(\alpha,j) = \infty$. 

We call an instance $(\alpha,j)$ \emph{extraneous} 
if the number of distinct release times of the jobs $[j]$ is at most $\alpha$, otherwise the instance is called \emph{non-extraneous}. Note that if there are no release times with more than $B$ jobs released at that time, and there are at most $\alpha$ distinct release times, then we can place an active slot at each distinct release time and schedule each job at its release time. The cost of this schedule is $j$.
When there are no jobs to be scheduled in an instance, notationally we let $j=0$. We define a \emph{B-capacity compatible} instance as one in which no more than $B$ jobs have the same release time, and will show that any instance of jobs and their accompanying release times can either be transformed into an equivalent \emph{B-capacity compatible} instance, or is infeasible.

Now, we can start discussing the structure of optimal solutions in this setting. As a reminder, we begin with a set of jobs that have been sorted first by release time and then by deadline. We will refer to the index of a job as its position in this ordering. We will show that a set of jobs and their accompanying release times can be pre-processed into an equivalent set of jobs that is $B$-capacity compatible.
Intuitively, this can be done incrementally by shifting the release times to the right so that no time slot ever contains more than $B$ release times.

The following lemma includes a slight abuse of notation. An instance $(\alpha,j)$ is inherently equipped with release times and deadlines. We alter the release times of jobs, keeping the rest of the jobs' information the same, and refer to the altered instance as $(\alpha,\widetilde{j}${} $)$, where the use of {} $\widetilde{\cdot}$ {} indicates that the jobs have new release times.

\begin{lemma}\label{preprocessing}
Let $(\alpha,j)$ be an instance of unit length jobs with agreeable deadlines.  In time $O(j)$, $(\alpha,j)$ can be transformed into a $B$-capacity compatible instance  $(\alpha,\widetilde{j}${} $)$ of the same number of unit jobs but with release times such that \textbf{(i)} at most $B$ jobs have the same release time and \textbf{(ii)} a schedule for $(\alpha,\widetilde{j}${} $)$ with minimal flow time is also a schedule for $(\alpha,j)$ with minimal flow time, with  $(\alpha,\widetilde{j}${} $)$ being infeasible if and only if  $(\alpha,j)$ is infeasible.
\end{lemma}
\begin{proof}
Fix instance $(\alpha,j)$.
We shift the release times of jobs from left to right such that all time slots we have passed over so far respect the capacity constraint. For $t \geq 1$, let $N_t$ be the number of jobs with release time $t$.
If $N_1 \leq B$ jobs, move to time slot 2.
Otherwise $N_1>B$ jobs are released at time $1$; take the $N_1-B$ jobs of largest index with release time 1, and move their release time to time slot 2, updating $N_2 \leftarrow N_2 + N_1-B$. Now time slot 1 has no more than $B$ jobs released there.
Continue this for time slot $t >1$, 
leaving the time slot alone if it has no more than $B$ jobs released there, and otherwise incrementing the release time by 1 for the $N_t-B$ largest indexed jobs with release time $t$. If the release time of a job $j$ is incremented past its deadline, then there must have been some interval of length $\ell$ such that more than $\ell \cdot B$ jobs needed to be completed within that interval, implying that the original instance was infeasible. The process terminates when we hit a time $t$ such that $t$ has no more than $B$ jobs and and no job has release time greater than $t$.
Let this new instance be $(\alpha,\widetilde{j}${} $)$, where jobs are equipped with their new release times and their same deadlines.

From construction, it is clear that in $(\alpha,\widetilde{j}${} $)$ no more than $B$ jobs are released at any given time, so $(\alpha,\widetilde{j}${} $)$ is indeed $B$-capacity compatible. Any schedule feasible/ infeasible for $(\alpha,\widetilde{j}${} $)$ is also feasible/ infeasible for $(\alpha,j)$, as the release times only increased from $(\alpha,j)$ to $(\alpha,\widetilde{j}${} $)$ and---as mentioned above---if a release time increased to beyond its deadline, then there were too many jobs that needed to be scheduled in $[r_j,d_j]$ in the original instance. Additionally, since no more than $B$ jobs can be scheduled in a time slot anyway, a feasible schedule with minimal flow time for $(\alpha,\widetilde{j}${} $)$ is also optimal for $(\alpha,j)$.
\end{proof}

We will assume the above pre-processing is done on the input of $n$ jobs, thus all $(\alpha,j)$ sub-problems are assumed to be $B$-capacity compatible in only $O(n)$ time.

Any valid schedule for jobs $[j]$ using $\alpha$ time slots will choose a job set of size at most $B$ that includes job $j$ and schedule this set in 1 time slot, while scheduling the remaining jobs in at most $\alpha-1$ time slots. The next lemma shows that among the schedules with minimum flow time, for feasible $(\alpha,j)$ there exists a schedule $S \in \mathcal{S}(\alpha,j)$ such that the order of the jobs corresponds to the order of the times at which the jobs are scheduled. We call such a schedule $S$ \emph{ordered}.

\begin{lemma}\label{staircase}
Let $(\alpha,j)$ be a feasible instance of  uniform jobs with agreeable deadlines. There exists an
 ordered schedule $S \in \mathcal{S}(\alpha,j)$, i.e.
for all $i,\ell \in [j]$ with $i < \ell$, if job $i$ is scheduled at time $t_i$ and job $\ell$ is scheduled at time $t_{\ell}$ in $S$, then $t_i \leq t_{\ell}$. It follows that the set of jobs scheduled in the same time slot as $j$ in $S$ is of the form
$[b+1,j]$ for $b \in [ j-B,j-1 ]$.
\end{lemma}

\begin{proof}
Take any schedule $S' \in \mathcal{S}(\alpha,j)$. 
If we cannot return $S \leftarrow S'$ as an ordered schedule,
then there are jobs $i$ and $\ell$, where $i < \ell$ is scheduled at time $t_i >t_{\ell}$ in $S'$. Let $S'^{(1)}$ be the schedule obtained from taking $S'$ and swapping $i$ and $\ell$, i.e. scheduling $i$ at time $t_{\ell}$ and $\ell$ at time $t_i$. 
First note that this $S'^{(1)}$ is feasible, as $i$ might move to an earlier time slot, but $t_{\ell}$ is feasible for $i$ since it was feasible for $\ell$ and $r_{i} \leq r_{\ell}$. Similarly, $\ell$ might move to a later time slot, but $t_i$ is feasible for $\ell$ as $d_i \leq d_\ell$.
Additionally, the flow time did not increase after swapping the two jobs.
If there are no more pairs of jobs whose scheduled times are out of order from their index, return $S \leftarrow S'^{(1)}$. Otherwise, swap a pair of jobs that are out of order to obtain $S'^{(2)}$. Continuing this for at most $j^2$ swaps will result in an ordered schedule $S \in \mathcal{S}(\alpha,j)$.

Since each time slot has capacity at most $B$ and $S$ is ordered, the set of jobs scheduled at the same time slot as $j$ in $S$ is of the form
$[b+1,j]$ for $b \in [j-B,j-1 ]$.
\end{proof}

A version of the previous lemma was also observed by Baptiste. See Proposition 1 in Baptiste's paper \cite{Baptiste00}.

Next, we work towards our main structural lemma for Theorem \ref{main}, which is Lemma \ref{j-release-time}. In part, this lemma proves that the active times will be a subset of the release times. It also shows that if $(\alpha,j)$ is $B$-capacity compatible, then there is no reason to schedule beyond time slot $r_j$.
The next three statements will be helpful in proving  Lemma \ref{j-release-time}.
Recall that $\textsf{OPT}(\alpha,j)$ is the flow time of any $S \in \mathcal{S}(\alpha,j)$.

\begin{claim}\label{less-jobs-opt}
Fix $\alpha \geq 1$.
For a feasible instance $(\alpha,j)$ of unit length jobs,
$\textsf{OPT}(\alpha,i) \leq \textsf{OPT}(\alpha,j)- (j-i)$ for all $i \leq j$.
\end{claim}
\begin{proof}
Consider a schedule $S \in \mathcal{S}(\alpha,j)$. Let $S'$ be $S$ restricted to jobs $[i]$. The cost of $S'$ is at most $\textsf{OPT}(\alpha,j) - (j-i)$, as each job in $[i + 1, j]$ contributes flow time at least 1 in $S$. On the other hand, since $S'$ is a feasible schedule for $[i]$ using at most $\alpha$ time slots, the cost of $S'$ is at least $\textsf{OPT}(\alpha,i)$.
\end{proof}

In our next few statements and proofs, it will be helpful to index the active time slots as $\tau_1,\ldots,\tau_{\alpha}$ for a non-extraneous instance.

\begin{claim}\label{last-time-slot}
Fix $\alpha >1$.
For a feasible, non-extraneous instance $(\alpha,j)$ of unit length jobs, 
let $S \in \mathcal{S}(\alpha,j)$ be a schedule with $\tau_{\alpha-1} < r_j$. Then $S$ has its last active time slot $\tau_\alpha$ at $ r_j$. 
\end{claim}
\begin{proof}
Assume for the sake of contradiction that some schedule $S \in \mathcal{S}(\alpha,j)$ satisfies all conditions of the claim but has $\tau_\alpha > r_j$. Let $J_{\alpha} \subseteq [j]$ denote the set of jobs assigned to $\tau_\alpha$ in $S$. 
Consider the schedule $S'$ that schedules all jobs as in $S$, except jobs in $J_\alpha$ are scheduled at $r_j$ instead of $\tau_\alpha$. $S'$ is clearly feasible, as $r_j$ is the latest release time of any job in $J_{\alpha}$ and $|J_{\alpha}| \leq B$, since $S$ was a valid schedule. The flow time of $S'$ is exactly $|J_\alpha| \cdot (\tau_\alpha - r_j)$ lower than that of $S$, contradicting the fact that $S \in \mathcal{S}(\alpha,j)$.

\end{proof}

\begin{claim}\label{less-slots-opt}
Fix $\alpha > 1$.
For a feasible, non-extraneous instance $(\alpha,j)$ of unit length jobs, if $\alpha' < \alpha$ then
$\textsf{OPT}(\alpha,j) < \textsf{OPT}(\alpha',j)$.
\end{claim}
\begin{proof}
Consider a schedule $S' \in \mathcal{S}(\alpha',j)$. For $i \in [j]$, pick any release time $r_i$ that is not currently in the set of active time slots. We know there is at least one such $r_i$, as $(\alpha, j)$ is non-extraneous. Add an active time slot $\tau$ at $r_i$. Create the schedule $S$ by scheduling all jobs with release time $r_i$ at $\tau$ (or only $B$ of them if more than $B$ jobs have release time $r_i$) and scheduling all other jobs as they were in $S'$. Any job scheduled at $\tau$ in $S$ has less cost than in $S'$ as it was moved strictly earlier, and any job scheduled not at $\tau$ has the same cost in $S'$ and $S$. Therefore the cost of $S$ is less than that of $S'$, which implies $\textsf{OPT}(\alpha,j) < \textsf{OPT}(\alpha',j)$.
\end{proof}

Equipped with the previous claims, we are ready to prove Lemma~\ref{j-release-time}.

\begin{lemma}\label{j-release-time}
For a feasible, $B$-capacity compatible, non-extraneous instance $(\alpha,j)$ of unit length jobs with agreeable deadlines, every ordered $S \in \mathcal{S}(\alpha,j)$
has its last active time slot $\tau_\alpha$ at $r_j$.
\end{lemma}
\begin{proof}
Note that since the instance is non-extraneous, a feasible schedule must have $\tau_\alpha \geq r_j$. The proof follows by induction on $\alpha$. \\
\\
\noindent
\textbf{Base Case:} Let $\alpha = 1$. Since the instance is feasible and $B$-capacity compatible, a single active slot at the only release time schedules the jobs optimally. \\ 
\\
\noindent 
\textbf{Inductive Step:} 
Fix an ordered $S \in \mathcal{S}(\alpha,j)$, which exists by Lemma \ref{staircase}. 
Observe that by the second statement in Lemma~\ref{staircase},
the set of jobs scheduled at $\tau_\alpha$ must be a consecutive set of the form $[b + 1, j]$ where  $b \in [j-B,j-1 ]$.
$S$ can be therefore represented by the union of two schedules $S_1, S_2$ such that $S_1$ schedules the first $b$ jobs in the first $\alpha - 1$ active time slots and $S_2$ schedules the remaining $j - b$ jobs in the last active time slot. 
If $S_1$ is extraneous, it must be that $r_b < r_j$ since $(\alpha,j)$ was non-extraneous. Here, every job in $S_1$ could be scheduled at their release times, which are before $r_j$. From Claim \ref{last-time-slot}, $\tau_\alpha=r_j$.
Therefore, the interesting case is when $S_1$ is also non-extraneous, which we will assume for the rest of the proof.
Note that $S_1$ indeed uses $\alpha-1$ active slots, as the instance is non-extraneous and so Claim \ref{less-slots-opt} applies.

Note that for any given $b$, the optimal cost of scheduling jobs $[b]$ in $\alpha - 1$ time slots is exactly $\textsf{OPT}(\alpha - 1, b)$, and the schedules that achieve this cost are in $\mathcal{S}(\alpha - 1, b)$. By the inductive hypothesis, we know that the last slot of $S_1$ is scheduled at $r_b$. Note that $r_b \leq r_j$ because the jobs are ordered by non-decreasing release time.

We now consider every possible $b \in [j-B,j-1 ]$ for a given $j$. We first consider $b'$ such that $r_{b'} < r_j$. Then by Claim \ref{last-time-slot}, we must have $\tau_\alpha = r_j$. The cost of schedule $S'$ based on choosing $b'$ is $\textsf{OPT}(\alpha - 1, b')+c_{b'}$ where $c_{b'}$ is the cost of scheduling jobs $[b' + 1, j]$ at $r_j$. 
There must be at least one $b'$ such that $(\alpha-1,b')$ is feasible and $r_{b'} < r_j$, namely $b' = j-B$. 
This is because the instance is $B$-capacity compatible with at least job $j$ having release time $r_j$, so there are at most $B - 1$ jobs directly before $j$ with release time $r_j$. Choose such a $b'$ with maximal index and let this be $b^*$, with corresponding schedule $S^*$ on $(\alpha, j)$. 

Now consider $b''$ such that $r_{b''} = r_j$. Then it must be the case that $\tau_\alpha > r_j$ because $\tau_{\alpha - 1}$ must be at $r_{b''}$ by the IH. The cost of schedule $S''$ based on choosing $b''$ is $\textsf{OPT}(\alpha - 1, b'') + c_{b''}$ where $c_{b''}$ is the cost of scheduling jobs $[b''+1, j]$ at $\tau_\alpha$.

For any such $b''$ and for our specially chosen $b^*$, $b^* < b''$, and so we can apply Claim~\ref{less-jobs-opt} to see
that $\textsf{OPT}(\alpha - 1, b^*) \leq \textsf{OPT}(\alpha - 1, b'') - (b''-b^*)$. We also know that all jobs scheduled at $\tau_\alpha$ in $S^*$ are scheduled at their release time, because all jobs in $[b^* + 1, j]$ have release time $r_j$. No jobs at $\tau_\alpha$ in $S''$ are scheduled at their release time, which implies $c_{b^*} < c_{b''} +(b''-b^*)$.
Adding these two inequalities, we see that
$$\textsf{OPT}(\alpha - 1, b^*) + c_{b^*} < \textsf{OPT}(\alpha - 1, b'') + c_{b''},$$
which shows that the optimal schedule chooses $b$ such that $r_b < r_j$. Thus by the inductive hypothesis, Claim~\ref{last-time-slot}, and the fact that $(\alpha,j)$ is non-extraneous, $S$ has its last active slot at $r_j$. \\
\noindent 
\end{proof}

\begin{figure}
    \centering
    \includegraphics[width = 8cm]{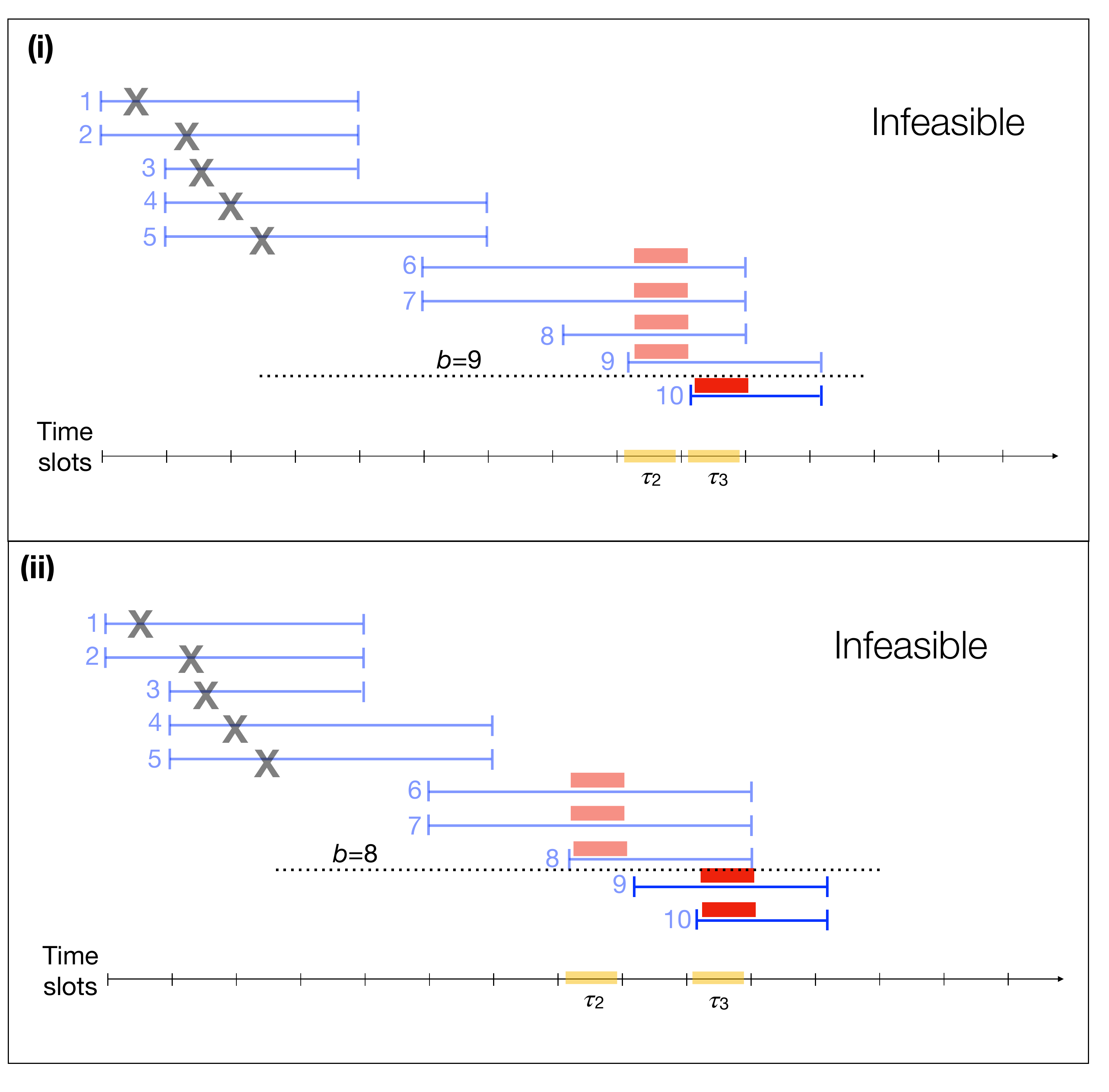}\\ 
        \caption{Blue intervals indicate release time and deadline intervals. Red rectangles indicate where a job is scheduled, and highlighted slots are active. Suppose $B=4$ and $\alpha = 3$.  The DP loops over choices of $b$.  Jobs $[b+1,j]$ are scheduled in $\tau_3$, then the DP recurses over the remaining subproblem.  
        Frame (i) scheduled job 10 in $\tau_3$, but it is impossible to schedule the 9 remaining jobs in 2 time slots. Frame (ii) tried to schedule 9,10 in $\tau_3$, but again the remaining subproblem is infeasible. Frame (iii) and frame (iv) provide feasible solutions.}
     \includegraphics[width = 8cm]{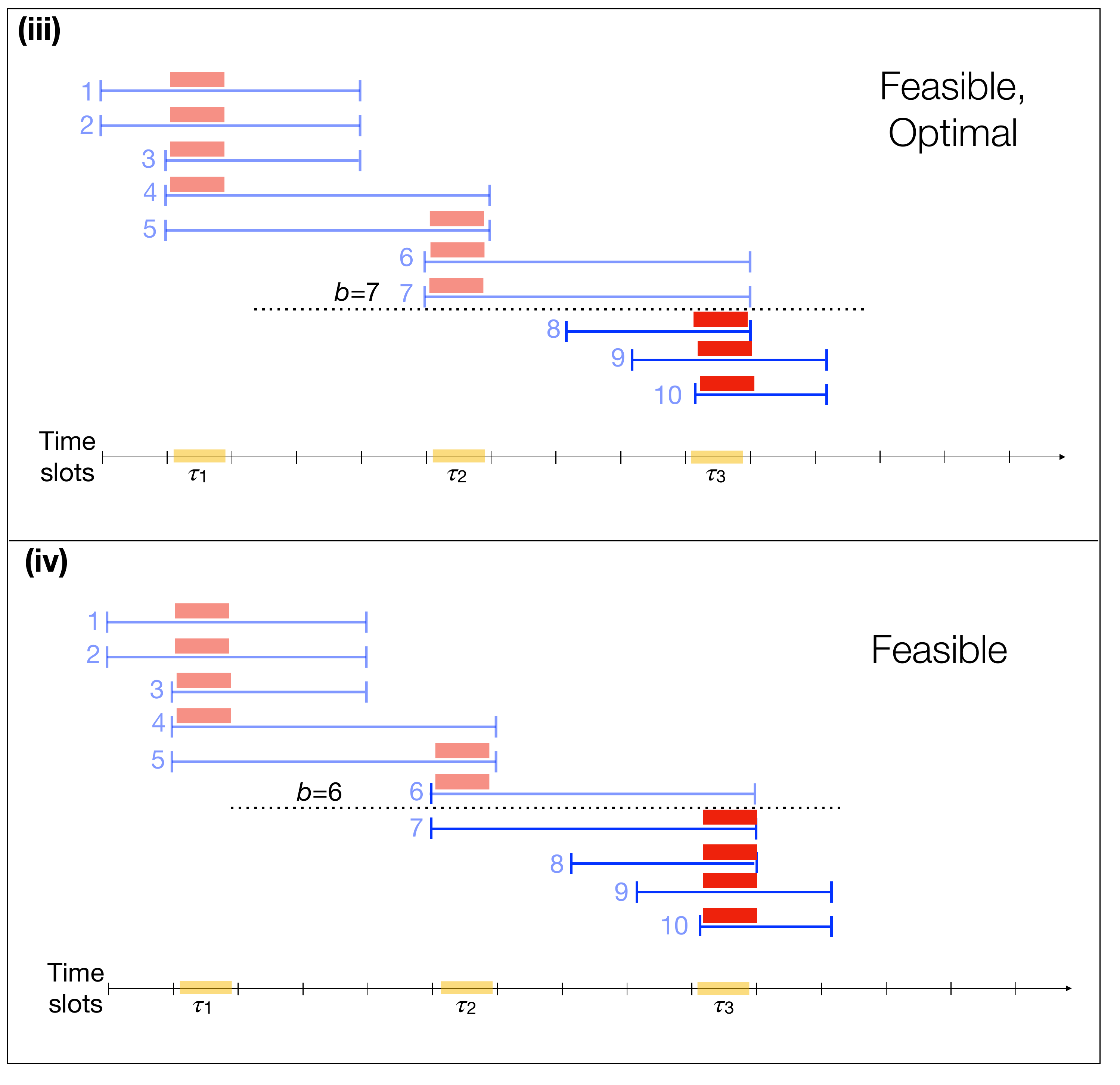}
    \label{fig: DP-ex}
\end{figure}

We are now ready to prove Theorem \ref{main}.

\begin{proof}[Proof of Theorem \ref{main}]
Let $(k,n)$ be an instance of unit length jobs with agreeable deadlines, where jobs are ordered by release times.
By Lemma \ref{preprocessing}, it is without loss of generality to assume that $(k,n)$ is $B$-capacity compatible.
For a fixed job $\ell$, let $i_\ell$ be the largest job with $r_{i_\ell} < r_\ell$.  
Fix any $0 \leq \alpha \leq k$ and $0 \leq j \leq n$.

Our dynamic program for this setting is the following:\\
   If $(\alpha,j)$ has $B \cdot \alpha < j$,  then  $\textsf{OPT}(\alpha,j) = \infty$. \\
    If  $(\alpha,j)$ has $B \cdot \alpha \geq j$ and is extraneous, then $\textsf{OPT}(\alpha,j) =j.$\\
    If $(\alpha,j)$ has $B\cdot \alpha\geq j$ and is non-extraneous, then:
    $$\textsf{OPT}(\alpha,j) =  \min_{\underset{ d_{b+1} > r_j}{b \in [j-B,i_j]:}} \left ( \textsf{OPT}(\alpha-1,b) +  \sum_{u=b+1}^j \left (r_j-r_{u} +1 \right) \right ).$$

We begin with the base cases of our dynamic program. If $(\alpha,j)$ has $B \cdot \alpha < j$,  then jobs $[j]$ cannot feasibly be scheduled, so  $\textsf{OPT}(\alpha,j) = \infty$.
If  $(\alpha,j)$ has $B \cdot \alpha \geq j$ and is extraneous, then every job can be scheduled at its release time, and so $\textsf{OPT}(\alpha,j) =j.$

The remaining setting outside of the base cases is when $B\cdot \alpha\geq j$ and is non-extraneous. We will prove the recurrence of our DP.

Suppose first that $(\alpha,j)$ is feasible.
If a base case does not apply, it must be that $(\alpha,j)$ is non-extraneous.
From Lemma~\ref{staircase}, we know that there is an optimal schedule that is ordered, which implies the last active slot consists of jobs $[b + 1, j]$ for some $b \in [j - B, i_j]$ as shown in Lemma \ref{j-release-time}.
The rest of the jobs use the remaining $\alpha-1$ active slots. In particular, no less than $\alpha-1$ slots are used as shown by Claim \ref{less-slots-opt}.
Additionally, Lemma \ref{j-release-time} shows that jobs in $[b+1,j]$ are scheduled at time slot $r_j$.  Since $b$ has $r_b < r_j$, jobs in $[b]$ will use no active time slot later than $r_b$, by Lemma \ref{j-release-time} and Lemma \ref{staircase}. 
It follows that an optimal schedule has flow time 
 $\textsf{OPT}(\alpha-1,b)$ contributed by the first $b$ jobs and flow time $ \sum_{u=b+1}^j \left (r_j-r_{u} +1 \right)$
 contributed by the rest of the jobs.
Taking the minimum over all possible $b$ gives the recurrence; see Figure \ref{fig: DP-ex}.
The proof of Lemma~\ref{j-release-time} also shows us that for feasible $(\alpha,j)$, there is always a choice for $b$ such that $(\alpha-1,b)$ is feasible. 
Finally, the number of active slots decreases by 1 each time, so the subroutines all terminate with some subset of the base cases.

Now suppose that $(\alpha,j)$ is infeasible. 
If $\alpha \cdot B < j$, then  $\textsf{OPT}(\alpha,j) = \infty$. Otherwise, we look at the recurrence.
Assume for sake of deriving a contradiction that there was some choice of $b$s, $\{b_1,\ldots,b_v\}$ in the recurrence that led to no subroutine  $(\alpha,j)$ with $B \cdot \alpha < j$. Jobs $[b_1+1,\ldots,j]$ are grouped into the $\alpha$ time slot, jobs $[b_2+1, b_1]$ are grouped in the $\alpha-1$st time slot, and so on until the base cases are hit, which occurs since $\alpha$ decreases monotonically. If only the base cases with finite cost are subschedules, then the recurrence has given rise to a feasible schedule for jobs $[j]$ that uses only $\alpha$ time slots. This contradicts the fact that no such schedule can exist for $(\alpha,j)$. 
Therefore, we can still identify infeasible instances by running the dynamic program which will output $\infty$ for $\textsf{OPT}(k, n)$ when the instance is infeasible.

The runtime of the dynamic program is $O(B \cdot k \cdot n)$, and the space complexity of the DP is $O(k \cdot n).$ The runtime and space complexity includes iterating over all values of $\alpha$ from $1$ to $k$ and jobs from $1$ to $n$ in the parameters of \textsf{OPT}. In addition, the inner minimization of $b$ contributes an extra factor of $B$ to the runtime. Note that the space complexity can be improved to $O(n)$ if we are only interested in the cost of the optimal schedule (and not the schedule itself), because in order to calculate $OPT(\alpha, j)$ for all $j < n$ we require only the values of $OPT(\alpha - 1, b)$ for all $b < n$.
Observe that the DP constructs a schedule obtaining the minimum flow time if $(k,n)$ is feasible. Otherwise, it gives a certificate that there is no feasible schedule.
\end{proof}

\subsection*{Completing a subset}\label{sec: thm2}

Here, we show that the dynamic program in Theorem  \ref{main} can be extended to the case when instead of completing all $n$ jobs, we only want to complete $m < n$ jobs.
We consider $m$ as part of the input, and we wish to choose the $m$ jobs that will be completed such that flow time is minimized out of all $\binom{n}{m}$ possible choices. We write such an instance as $(k,n)_m$,
 and denote subproblems as $(\alpha,j)_q$, for $0 \leq \alpha \leq k$, $0 \leq j \leq n$, and $ 0 \leq q \leq m,j$.

The instance $(\alpha,j)_q$ is feasible exactly when it is possible to schedule $q$ jobs of $[j]$ in $\alpha$ active time slots. $\textsf{OPT}(\alpha,j)_q$ is the minimum flow time achievable for $(\alpha,j)_q$. When $(\alpha,j)_q$ is infeasible, $\textsf{OPT}(\alpha,j)_q = \infty$.
For feasible $(\alpha,j)_q$, $\mathcal{S}(\alpha,j)_q$ denotes the set of schedules with flow time $\textsf{OPT}(\alpha,j)_q$ that use $\alpha$ active slots and schedule $q$ of the jobs in $[j]$.
We also call an instance $(\alpha,j)_q$ extraneous if there exists a choice of $q$ jobs from $[j]$ such that the number of distinct release times of those $q$ jobs is at most $\alpha$, otherwise the instance is called non-extraneous. 

Overall, we show the following extension.

\begin{thm}\label{subset}
Let $(k,n)$ be an instance of unit length jobs with agreeable deadlines. Fix $0 \leq m \leq n$. Then one can either certify $(k,n)_m$ is infeasible or find a schedule for $(k,n)_m$ that minimizes the flow time with a dynamic program in time $O(B \cdot k \cdot n \cdot (n-m))$ and space $O(k \cdot n \cdot (n-m))$.
\end{thm}

Again, we can still pre-process the jobs so that no more than $B$ have the same release time, as in Lemma~\ref{preprocessing}; the same proof holds.
There still exists an ordered optimal schedule---i.e. a schedule $S \in \mathcal{S}(\alpha,j)_q$ such that if $i < \ell$ and $i$ and $\ell$ are scheduled, then $t_i \leq t_\ell$, as in Lemma~\ref{staircase}. The proof that an ordered optimal solution exists still holds, but our definition of ordered is not quite strong enough here. Additionally, there exists an ordered $S \in\mathcal{S}(\alpha,j)_q$ such that if $j$ is scheduled in $S$, the set of jobs scheduled at the same time slot as $j$ in $S$ is of the form $[b+1,j]$ for $b \in [j-B,j-1].$ To see this, consider any ordered $S'  \in\mathcal{S}(\alpha,j)_q$ and suppose that $j$ was scheduled with a non-contiguous indexed set of jobs in $S'$ of size $\ell \leq B$. Then let $S$ be a schedule that assigns all jobs as in $S'$ except potentially the $\ell$ jobs scheduled with $j$ in $S'$; schedule jobs $[j-\ell,j]$ with $j$ instead. $S$ and $S'$ schedule $q$ jobs, they are both ordered, and $S$ does not have greater flow time than $S'$ since jobs have agreeable deadlines. It follows that $S \in \mathcal{S}(\alpha,j)_q$ has the desired property. However, it is not hard to see that one can do this for every active time slot, not just the one that $j$ is assigned to. So there exists a schedule in $ \mathcal{S}(\alpha,j)_q$ such that for every active time slot $\tau_\beta$ for $1 \leq \beta \leq \alpha$, the set of jobs scheduled at $\tau_\beta$ is a contiguous set of indexed jobs $J_\beta = [l_\beta, u_\beta]$, where $u_\beta - 
l_\beta \leq B$.
We refer to such a schedule as being ordered and \emph{contiguous}, as every active time slot schedules a contiguous interval of jobs. 
Note that when we were forced to complete all jobs, ordered schedules were also contiguous. We state the analogous version of Lemma \ref{staircase} for the subset setting, whose proof follows from the discussion above.

\begin{lemma}\label{staircase-subset}
Let $(\alpha,j)_q$ be a feasible instance of  uniform jobs with agreeable deadlines. There exists an
 ordered, contiguous schedule $S \in \mathcal{S}(\alpha,j)_q$. In other words,
for all $i,\ell \in [j]$ with $i < \ell$, if job $i$ is scheduled at time $t_i$ and job $\ell$ is scheduled at time $t_{\ell}$ in $S$, then $t_i \leq t_{\ell}$, and the set of jobs $S$ assigns to any time slot $\tau_\beta$ is a contiguous set of indices $[l_\beta, u_\beta]$, where $u_\beta - 
l_\beta \leq B$.
\end{lemma}

Overall, we can show a structural lemma very similar to that of Lemma \ref{j-release-time}. However, we will need more general versions of Claims \ref{less-jobs-opt}, \ref{last-time-slot}, and \ref{less-slots-opt} in order to prove that structural lemma. The proofs of these claims are omitted, as they are only slight, simple variations of the original claims.

\begin{claim}\label{less-jobs-opt-subset}
Fix $\alpha \geq 1$.
For a feasible instance $(\alpha,j)_{j-s}$ of unit length jobs,
$\textsf{OPT}(\alpha,i)_{i-s} \leq \textsf{OPT}(\alpha,j)_{j-s}- (j-i)$ for all $i \leq j$ and $s \geq 0$.
\end{claim}

\begin{claim}\label{last-time-slot-subset}
Fix $\alpha >1.$
For a feasible, non-extraneous instance $(\alpha,j)_q$  of unit length jobs, let $S \in \mathcal{S}(\alpha,j)_q$ be an ordered schedule with $\tau_{\alpha-1} < r_{j'}$, for $j'\leq j$ the largest indexed job scheduled by $S$. Then $S$ has its last active time slot $\tau_\alpha$ at $ r_{j'}$. 
\end{claim}

\begin{claim}\label{less-slots-opt-subset}
Fix $\alpha >1.$
For a feasible, non-extraneous instance $(\alpha,j)_q$ of unit length jobs, if $\alpha' < \alpha$ then
$\textsf{OPT}(\alpha,j)_q < \textsf{OPT}(\alpha',j)_q$.
\end{claim}

Now we state our main structural lemma for the setting when we only complete a subset of jobs. We omit some details of the proof as it is very similar to that of Lemma \ref{j-release-time}.

\begin{lemma}\label{j-release-time-subset}
For a feasible, $B$-capacity compatible, non-extraneous instance $(\alpha,j)_q$ of unit length jobs with agreeable deadlines, an ordered, contiguous $S \in \mathcal{S}(\alpha,j)_q$
has its last active time slot $\tau_\alpha$ at $r_{j'}$, for $j'$ the largest indexed job in $[j]$ scheduled in $S$.
\end{lemma}
\begin{proof}
 The proof follows by induction on $\alpha$; the base case is the same as Lemma \ref{j-release-time}.

Fix an ordered, contiguous $S \in \mathcal{S}(\alpha,j)_q$, 
which exists by Lemma \ref{staircase-subset}.
The set of jobs scheduled by $S$ at $\tau_\alpha$ must be a consecutive set of the form $[b + 1, j']$, where  $b \in [j'-B,j'-1 ]$ and $j' \leq j$ is the highest indexed job scheduled in $S$.
For $b'$ with $r_{b'} < r_{j'}$, Claim \ref{last-time-slot-subset} implies that $\tau_\alpha = r_{j'}$. 
Claim \ref{less-slots-opt-subset} implies that $\alpha-1$ active slots are better than $\alpha' < \alpha-1$ to schedule the remaining cost,
so the cost of such a schedule is $\textsf{OPT}(\alpha - 1, b')_{ q-(j'-b')}+c_{b'}$ where $c_{b'}$ is the cost of scheduling jobs $[b' + 1, j']$ at $r_{j'}$. 
Let $b^*$ be of maximal index so that $r_{b^*} < r_{j'}$ and
$(\alpha - 1, b^*)_{q-(j'-b^*)}$ is feasible.
Now consider $b''$ such that $r_{b''} = r_{j'}$. Then it must be the case that $\tau_\alpha > r_{j'}$. The cost of a schedule based on choosing $b''$ is $\textsf{OPT}(\alpha - 1, b'')_{ q-(j'-b'')} + c_{b''}$ where $c_{b''}$ is the cost of scheduling jobs $[b'', j']$ at $\tau_\alpha$. 
For any such $b''$ and for our specially chosen $b^*$, we see that  $b^* < b''$ and so we can apply Claim~\ref{less-jobs-opt-subset} to see
that 
$$\textsf{OPT}(\alpha - 1, b^*)_{ q-(j'-b^*)}\leq \textsf{OPT}(\alpha - 1, b'')_ {q-(j'-b'')} - (b''-b^*).$$ 
Further, $c_{b^*} < c_{b''} +(b''-b^*)$.
Adding these two inequalities,
$$\textsf{OPT}(\alpha - 1, b^*)_{q-(j'-b^*)} + c_{b^*} < \textsf{OPT}(\alpha - 1, b'')_{ q-(j'-b'')} + c_{b''},$$
which shows the optimal schedule chooses $b$ such that $r_b < r_j$.
\noindent 
\end{proof}

Overall, the proof of correctness for the DP is the same as Theorem \ref{main}, just with the new analogous lemmas. The DP in its entirety is below.

\begin{proof}[DP proving Theorem \ref{subset}]
Let $(k,n)$ be an instance of unit length jobs with agreeable deadlines, where jobs are ordered by release times, and w.l.o.g. assume that $(k,n)$ is $B$-capacity compatible. Fix $0 \leq m \leq n$. For a fixed job $\ell$, let $i_\ell$ be the largest job with $r_{i_\ell} < r_\ell$.
We formally state the DP for  $0 \leq \alpha \leq k$, $0 \leq j \leq n$, and $0 \leq q \leq j,m$:\\
If $(\alpha,j)_q$ has $B \cdot \alpha < q$, then  $\textsf{OPT}(\alpha,j)_q = \infty$. \\
If  $(\alpha,j)_q$ is extraneous, then $\textsf{OPT}(\alpha,j)_q =q $.\\
If $(\alpha,j)_q$ has $B \cdot \alpha \geq q$ and is non-extraneous, then:
\begin{align*}
       \textsf{OPT}(\alpha,j)_q = \min_{ \ell \in [q,j] } \min_{\underset{ d_{b+1} > r_\ell}{b \in [\ell-B,i_\ell]}} \big ( &\textsf{OPT}(\alpha-1,b)_{ q-(\ell-b)} \\
     &+\sum_{u=b+1}^{\ell} \left (r_{\ell}-r_{u} +1 \right) \big ).
\end{align*}

As before, the DP will certify infeasible $(k,n)_m$ are infeasible.
For feasible $(k,n)_m$, the recurrence for the minimum flow time also constructs a feasible schedule obtaining that flow time. The runtime and space complexity now include an extra factor of $(n - m)$ from the minimization over $\ell$ which decides which jobs are not completed.
\end{proof}

\section{Uniform jobs, agreeable deadlines}\label{sec: uniform}

In this section, we prove Theorem \ref{uniform}, which applies to uniform jobs with length $p \in \mathbb{N}$ and agreeable deadlines. 
Here, there is less structure to use for the problem, which forces us to consider more possibilities for where to schedule active time slots. 
This leads to a runtime that is better than augmenting Baptiste's DP for arbitrary deadlines, but worse then our unit length setting. 
We assume the processor performs synchronous batching, where no batch can start until the last is finished, even if the previous batch was not full.

There are several obstacles to adapting our previous DP formulation for this case. First, it is no longer sufficient to start active batches only at the release times of jobs, because it may be the case that a release time falls within the processing time of another batch. 
It is also not obvious how to pre-process jobs such that they can be scheduled at their release times.
Second, it is no longer true that it is at least as good for flow time to use $\alpha$ non-empty active batches instead of $\alpha' < \alpha$ non-empty active batches. Recall that Claim \ref{less-slots-opt} showed that this was the case for unit jobs, assuming the number of unique release times in $[j]$ is at least $\alpha$. This is because the optimal start time of one batch may overlap the processing time of another batch, and so it is not always possible to start a new batch between two others, even if the active slot budget is increased. Overall, we forego pre-processing here and instead add an additional parameter, $t$, that keeps track of the first of the last $p$ time slots used by job $j$ in scheduling $[j]$. Similarly we forego the notion of extraneous here, as even if there are enough active slots to place one at every distinct release time, this does not necessarily lead to a feasible schedule, even if every time slot did have at most $B$ jobs released.

We say an instance $(\alpha,j,t)$ is feasible if all jobs in $[j]$ can be scheduled in at most $\alpha$ active batches (i.e., $\alpha \cdot p$ active time slots), where the last $p$ active time slots cover $[t,t+p)$.
For feasible $(\alpha,j,t)$, $\textsf{OPT}(\alpha,j,t)$ is the minimum flow time of a schedule with at most $\alpha$ active batches that schedules $[j]$ and has its last $p$ active slots in $[t,t+p)$; let $\mathcal{S}(\alpha,j,t)$ be the set of optimal schedules. If $(\alpha,j,t)$ is infeasible, then $\textsf{OPT}(\alpha,j,t)=\infty$.
We have that $\textsf{OPT}(\alpha,j) = \min_t \textsf{OPT}(\alpha,j,t)$, where if $(\alpha,j)$ if feasible, this quantity is finite, and infinite otherwise.
Recall jobs are ordered by release times.

Importantly, note that Lemma \ref{staircase} applies for uniform (not just unit) jobs. Therefore if $(\alpha,j,t)$ is feasible, then an ordered $S \in \mathcal{S}(\alpha,j,t)$ exists.

\begin{proof}[Proof of Theorem \ref{uniform}]
Let $(k,n)$ be an instance of length $p$ jobs with agreeable deadlines, where jobs are ordered by release times.
Let  $T = \{r_j + p\cdot u\}$ for $j\in [n]$ and $0 \leq u \leq n$.
Then the following holds for all $t \in T \cup \{\min(T)-p\}$, $0 \leq \alpha \leq k$, and $0 \leq j \leq n$:\\
If $j > 0$ and either $t < r_j$ or $t+p > d_j$, then $\textsf{OPT}(\alpha,j,t)=\infty$.\\
If $B \cdot \alpha  < j$, then $\textsf{OPT}(\alpha,j,t)=\infty$.\\
For $\alpha \geq 0$ and $t \geq -p+1$, $\textsf{OPT}(\alpha,0,t)=0$.\\
Otherwise,
\begin{align*}
\textsf{OPT}(\alpha,j,t) =  \hspace{-4mm} \min_{\underset{ d_{b+1} \geq t+p}{b \in [j-B,j-1]:}}   \min_{t' \leq t-p} & 
\bigg ( \hspace{-.5mm}\textsf{OPT}(\alpha-1,b,t')
 +\hspace{-2mm} \sum_{u=b+1}^j (t-r_{u}+p) \hspace{-.5mm}\bigg).
\end{align*}
Then
$\textsf{OPT}(\alpha,j) =  \min_t  \textsf{OPT}(\alpha,j,t)$.

We begin with the base cases of our dynamic program. Clearly the first case is infeasible as job $j$ cannot be scheduled feasibly.
If $B \cdot   \alpha  < j$, the active slots do not have enough volume to schedule $[j]$. In both cases,  $\textsf{OPT}(\alpha,j) = \infty$.
Fix $\alpha$ and $j$ such that $0 \leq \alpha \leq k$ and $0 \leq j \leq n$ such that neither of the base cases hold.

Suppose that $(\alpha,j)$ is feasible.
Fix an ordered $S \in \mathcal{S}(\alpha,j)$, which exists by Lemma \ref{staircase}.
Then there exists some $p$ time slots where job $j$ is scheduled in $S$. Let $t$ be the first of these slots, so we consider the subproblem $(\alpha,j,t)$. Job $j$ is scheduled in a job set $[b+1,j]$, for $b \in [j-B,j-1]$, in $S$. The rest of the jobs, $[b]$, use at most $\alpha-1$ batches. By the definition of ordered, jobs in $[b]$ will not use a time slot later than $t-1$ and jobs in $[b+1,j]$ only use time slots $[t,t+p)$. The optimal subschedule for jobs in $[b]$ has flow time $\textsf{OPT}(\alpha-1,b,t')$, for $b$ with $r_b \leq t' \leq t-p$, $d_{b+1} \geq t+p$, and $d_b \geq t' + p$.
The flow time of jobs $[b+1,j]$ is $ \sum_{u = b+1}^j (t-r_u+p)$. Taking the minimum over all choices of $b$ and $t'$ gives the recurrence. 

The number of active batches decreases in each subschedule, so it remains to see that the finite base case is reached. Let $\beta \leq \alpha$ be the number of non-empty active batches used by $S$. For $i \in [\beta]$, let $t_i$ be the time where the $((i-1) \cdot p+1)$st active slot time is and let $b_i$ be the largest indexed job scheduled at $t_i$. Note that $b_s=j$. Also, take $t_0= t_1-p$ and $b_0=0$, where the latter occurs since there are no jobs less than $b_1$ that are not scheduled in the same active time slots as $b_1$.
The sequence of subproblems $\textsf{OPT}( i+\alpha-\beta,b_i,t_i)$ terminates in $\textsf{OPT}(\alpha-\beta,0,t_1-p)$, where $t_1-p \geq 1-p$.

If $(\alpha,j)$ is infeasible, then as in the proof of Theorem \ref{main}, a base case of infinite cost is caught.
The runtime of the dynamic program is $O(B \cdot k \cdot n^5)$, and its space complexity is $O(k \cdot n^3).$ 
The runtime and space complexity includes iterating over all values of $\alpha$ from $1$ to $k$, jobs from $1$ to $n$, and values of $T$ with $|T| = n^2$ in the parameters of \textsf{OPT}. In addition, the inner minimization of $b$ and $t'$ contributes an extra factor of $B \cdot n^2$ to the runtime.
Again, the DP constructs a schedule obtaining the minimum flow time if $(k,n)$ is feasible, and returns $\infty$ otherwise.
\end{proof}

\subsection*{Completing a subset of uniform jobs}\label{sec: thm4}
In this section, we justify that the DP in Theorem \ref{uniform} can be extended to when only $m \leq n$ jobs must be completed.

\begin{thm}\label{uniform-subset}
Let $(k,n)$ be an instance of uniform length jobs with agreeable deadlines. Fix $0 \leq m \leq n$. Then one can either certify $(k,n)_m$ is infeasible or find a schedule for $(k,n)_m$ that minimizes the flow time with a dynamic program in time $O(B \cdot k \cdot n^5 \cdot (n-m))$ and space $O(k \cdot n^3 \cdot (n-m))$.
\end{thm}

As in the setting where we complete a subset of unit length jobs, for $(\alpha,j)_q$ feasible, the proof of existence of an ordered $S \in \mathcal{S}(\alpha,j)_q$ is the same as in Lemma \ref{staircase}, but we require the extra notion of contiguous. It is easy to see that the same procedure works, guaranteeing that in every active batch $\tau_\beta$, for $1 \leq \beta \leq \alpha$, the set of jobs scheduled in batch $\tau_\beta$ is a contiguous set of indexed jobs with size at most $B$. Note that $q$-subscript indexing is used anywhere we consider solutions where $q$ of the $j$ jobs should be selected.

\begin{proof}[Proof of Theorem \ref{uniform-subset}]
Let $(k,n)$ be an instance of length $p$ jobs with agreeable deadlines, where jobs are ordered by release times.
Fix $0 \leq m \leq n$.
For a fixed job $\ell$, let $i_\ell$ be the largest job with $r_{i_\ell} < r_\ell$. 
Let  $T = \{r_j + p\cdot u \}$ for $j\in [n]$ and $0 \leq u \leq n$. Then for all $t \in T$, $0 \leq \alpha \leq k$, $0 \leq j \leq n$, and $0 \leq q \leq m,j$ the following holds:\\
If $j >0$ and either $t < r_j$ or $t+p > d_j$, then $\textsf{OPT}(\alpha,j,t)_q=\infty$.\\
If $B \cdot  \alpha < q$, then $\textsf{OPT}(\alpha,j,t)_q=\infty$. \\
For $\alpha,j \geq 0$ and $t \geq -p+1$, $\textsf{OPT}(\alpha,j,t)_0=0$.\\

Otherwise,
\begin{align*}
\textsf{OPT}(\alpha,j,t)_q = \min_{\ell \in [q,j]} \min_{\underset{ d_{b+1} \geq t+p}{b \in [\ell-B,i_\ell]:}} \quad  \min_{t' \leq t-p} & 
\bigg (\textsf{OPT}(\alpha-1,b,t')_{q-(\ell-b)}\\
& + \sum_{u=b+1}^\ell (t-r_{u}+p)\bigg).
\end{align*}
Then
$\textsf{OPT}(\alpha,j)_q =  \min_t  \textsf{OPT}(\alpha,j,t)_q$.

The statements clearly hold except for the recurrence.
Fix such an $\alpha$, $j$, and $q$.
Suppose that $(\alpha,j)_q$ is feasible, and fix an optimal, ordered, contiguous schedule $S \in \mathcal{S}(\alpha,j)$.
Let $\ell$ be the highest indexed job in $[j]$ scheduled in $S$. Then there exists some $p$ time slots in $S$ where job $\ell$ is scheduled. Let $t$ be the first of these slots. Therefore job $\ell$ is scheduled with a job set $[b+1,\ell]$ for $b \in [\ell-B,\ell-1]$. The rest of the jobs, $[b]$, use at most $\alpha-1$ batches. Jobs in $[b]$ will not use a time slot later than $t-1$ and jobs in $[b+1,\ell]$ only use time slots $[t,t+p]$. The optimal subschedule for jobs in $[b]$ has flow time $\textsf{OPT}(\alpha-1,b,t')_{q - (\ell -b)}$, for $b$ with $r_b \leq t' \leq t-p$, $d_{b+1} \geq t+p$, $d_b \geq t' + p$, and $0 \leq \alpha' \leq \alpha-1$.
The flow time of jobs $[b+1,\ell]$ is $\sum_{u = b+1}^\ell (t-r_u+p)$. Taking the minimum over all choices of $b$, $t'$, and $\ell$ gives the recurrence. 

The rest of the proof follows exactly as in the proof of Theorem \ref{uniform}, except using that the separators begin with $b_s=\ell$.

The DP runs in time $O(B \cdot k \cdot n^5 \cdot (n-m))$ and space $O(k \cdot n^3 \cdot (n-m))$, and identifies infeasible $(k,n)_m$.

\end{proof}

\section{Non-agreeable Deadlines}\label{Baptiste-sec}

Here, we justify how Baptiste's framework can be modified to handle a budget of active time slots. We omit details, as the arguments almost exactly follow those presented by Baptiste. The reader can refer to Baptiste's work for details \cite{Baptiste00}.

Baptiste's DP keeps track of several more variables than the DPs we have presented so far. Recall $p$ is the length of the jobs. Deadlines are not required to be agreeable, and here, jobs are ordered by deadline. For uniform jobs, Baptiste justifies that it suffices to only consider optimal schedules that are ordered (where the definition of ordered is now with respect to deadlines instead of release times) and that have all jobs scheduled in the set $T$, where $T = \{r_j + p\cdot i\}_{j \in [n],i \in \{0,...,n\} }$. Note that $|T| = n^2$. We call the set of times where jobs can be scheduled the set of \emph{interesting times}. If a job is not scheduled at its release time, then it might be scheduled at another job's release time, or other batches of jobs are using the release times, and so we must wait for them to be completed before starting other batches. Since we want to minimize flow time, batches of the latter form start some multiple of $p$ after a release time.
We use the same algorithm for the unit and uniform jobs settings, but when jobs are unit length and $B$-capacity compatible, we can reduce the set of interesting times from $T$ to $T_1 = \{r_j\}_{j \in [n]}$. The set of interesting times thus reduces in size from $n^2$ to size $n$, which improves the runtime.  

\begin{lemma}
Fix $\alpha >1$.
For a feasible, $B$-capacity compatible instance $(\alpha,j)$  of unit length jobs, 
whose set of release times is $T_1$,  every optimal schedule $S \in \mathcal{S}(\alpha,j)$ has active slots only in $T_1$.
\end{lemma}

\begin{proof}
Assume for sake of deriving a contradiction that there is a schedule $S \in \mathcal{S}(\alpha,j)$ that has an active time slot outside of $T_1$. Choose the earliest such active slot and denote this slot by $\tau$. Let the set of jobs scheduled at $\tau$ be $J_\tau$. 

We can modify $S$ with the following algorithm. Let $\tau' \gets \max_{j \in J_\tau} r_j$ be the latest release time of any job in $J_\tau$. We shift all jobs from $\tau$ to $\tau'$, so the set of jobs scheduled at $\tau'$ and $\tau$ are updated to  $J_{\tau'} \gets J_{\tau'} \cup J_{\tau}$ and $J_\tau=\emptyset$. The flow time of $S$ would only decrease by such a shift. If $|J_{\tau'}| >B$, we need to continue shifting jobs to the left, which can be done by keeping the $B$ jobs with the latest release time scheduled at $\tau'$ and shifting the rest of the jobs in $J_{\tau'}$ to the latest release time available. This process terminates since the instance has no more than $B$ jobs with the same release time. 
\end{proof}

Our notation is chosen to match Baptiste's as much as possible.
We will consider times $t_l,t_r \in T$ in the uniform case and $t_l,t_r \in T_1$ in the unit case. For a specific $t_l, t_r$, we consider the interval between these endpoints. In this time interval, we consider the set of jobs below a certain index whose release time lies in this interval and let this set be $U_j(t_l,t_r) = \{j' \mid j' \leq j, r_{j'} \in (t_l,t_r]\}$. At the rightmost $p$ length slot in the interval of consideration, $[t_r,t_r+p)$, we keep track of the space available with parameter $0 \leq \mu_r \leq B$. The DP will call upon the sub-problems $\textsf{OPT}(t_l,t_r,\mu_r,\alpha,j)$, which gives the flow time of an optimal schedule for jobs in $U_{j}(t_l,t_r)$ using $\alpha$ active batches in the interval $[t_\ell+p,t_r+p)$ with slots $[t_r,t_r+p)$ having $\mu_r$ space available. 
In the uniform case, we can modify Baptiste's DP to produce a schedule with cost $\textsf{OPT}(\min(T) - p,\max(T) + p,0,k,n)$. Note that $\min(T) - p, \max(T) + p$ bookend all possible times at which jobs can be scheduled and all release times, so this represents the full problem. However, $\max(T) + p$ is not itself in $T$, so the space available at $\mu_r$ is $0$.
The same line of reasoning holds in the unit jobs setting for $\textsf{OPT}(\min(T_1) - 1,\max(T_1) + 1,0,k,n)$.

\begin{figure}
    \centering
    \includegraphics[width=8cm]{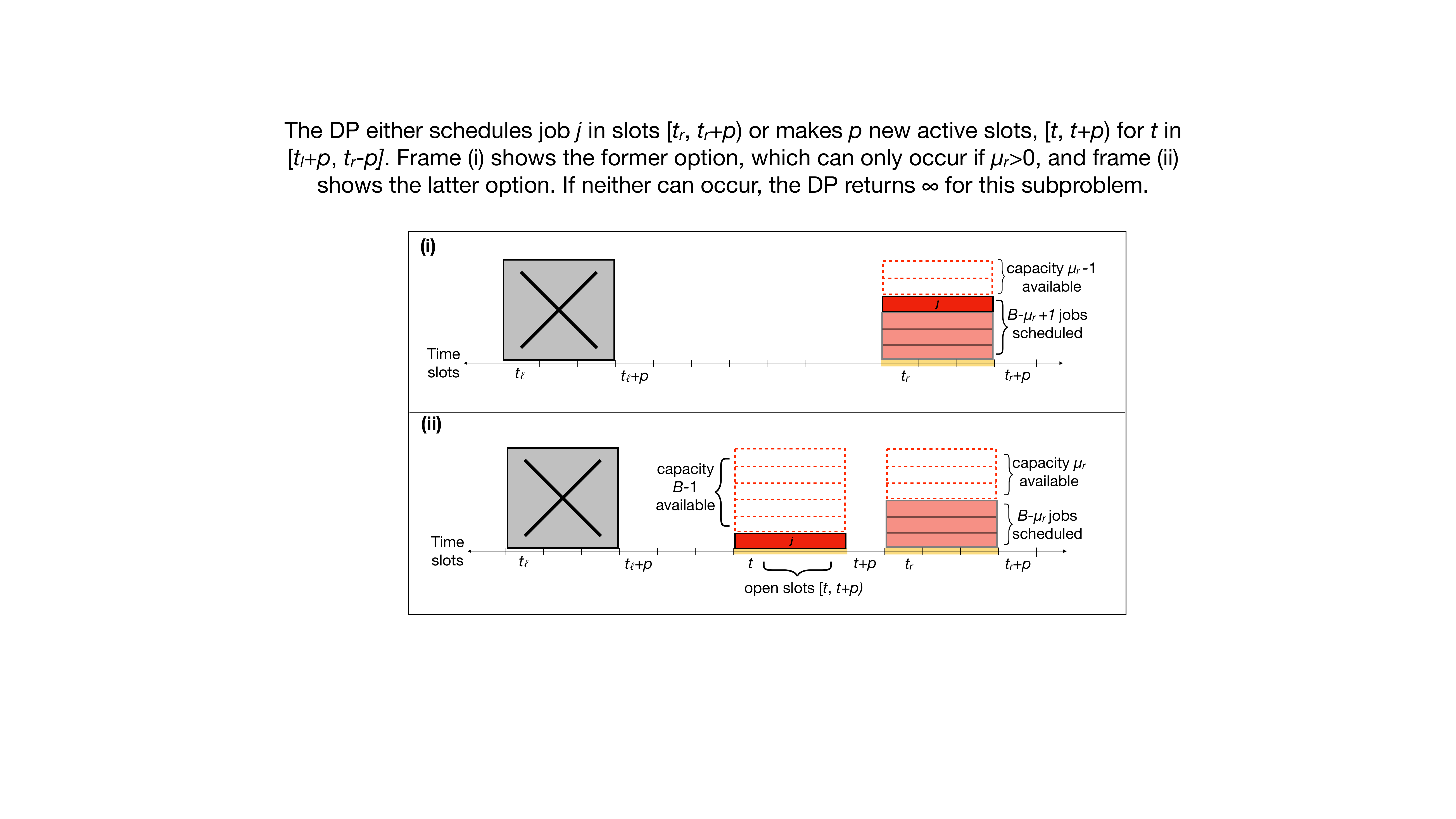}
    \caption{The DP either schedules job j in slots $[t_r, t_r+p)$ (this is subproblem $R$) or makes $p$ new active slots, $[t, t+p)$ for $t$ in $[t_\ell+p, t_r-p]$ (this is subproblem $I$). Frame (i) shows the former option, which can only occur if $\mu_r>0$, and frame (ii) shows the latter option. If neither can occur, the DP returns $\infty$ for this subproblem.}
    \label{fig:Baptiste}
\end{figure}

\begin{proof}[Proof sketch of Theorem \ref{thm:Baptiste-ext}]
We write the DP for uniform jobs. One can write the unit jobs DP by replacing $T$ with $T_1$ and $p$ with 1.

At each recursive step, the DP decides whether to begin processing job $j$ at the rightmost time slot in the active interval ($t_r$) or at an in-between time slot $t \in[t_l+p,t_r-p]$; the former will be referred to as subproblem $R$ and the latter as subproblem $I$. See Figure \ref{fig:Baptiste} for an illustration.
Recall that because jobs have length $p$, when a new time slot opens, we actually open up $p$ contiguous slots for the synchronous batching.
Note we do not consider scheduling job $j$ before $t_l+p$ because jobs from a subproblem with $ t_\ell -p < t' \leq t_{\ell}$ as its right endpoint may be using slots in $[t_\ell, t_{\ell}+p)$. The DP can only process $j$ in $[t_r,t_r+p)$, i.e. compute subproblem $R$ with finite cost, if there is space available. Additionally, in $I$ the DP must decide how to partition the $\alpha$ active batches between $[t_l+p,t+p)$ and $[t+p,t_r+p)$. For $\alpha_1$ the maximum number of active batches allowed in $[t_l+p,t+p)$ and $\alpha_2$  maximum number of active batches allowed in $[t+p,t_r)$, we must have that $\alpha_1 + \alpha_2 \leq \alpha$, and also $\alpha_1 \geq 1$, as job $j$ will be scheduled at slots $[t, t + p)$. 

Initialize $I=\infty$ and $R = \infty$. 
Overall, we have the following for  $t_l\leq t_r \in T$, $0 \leq \mu_r \leq B$, $j \in [n]$, $\alpha \in [k]$:
$$
\textsf{OPT}(t_l,t_r,\mu_r,\alpha,j) = \min (I,R), 
$$
where if $\mu_r=0$ then $R = \infty$ and otherwise
$$
R= \textsf{OPT}(t_l,t_r,\mu_r-1,\alpha,j-1) + (t_r+p - r_j),
$$
and for $T_j = T \cap [t_l + p,t_r-p] \cap [r_j,d_j]$---the set of time slots under consideration for beginning to schedule job $j$, given that $j$ is not scheduled in $[t_r,t_r+p)$
\begin{align*}
 I = \min  \bigg (I, \quad \min_{t \in T_j}  \min_{ \underset{\alpha_1+\alpha_2 = \alpha, 1 \leq \alpha_1\leq \alpha}{ \alpha_1,\alpha_2 :}} \textsf{OPT}(t_l,t,B-1,\alpha_1,j-1) +\\ (t+p-r_j) + \textsf{OPT}(t,t_r,\mu_r,\alpha_2,j-1)
 \bigg ).
\end{align*}

The runtime is $O(B \cdot k^2 \cdot n \cdot |T|^3)$ and the space complexity $O( B \cdot k \cdot n \cdot |T|^2)$. 
The runtime and space complexity includes iterating over all values of $t_l \in T$, $t_r \in T$, $\mu_r \in [B]$, $\alpha \in [k]$, and jobs from $1$ to $n$ in the parameters of \textsf{OPT}. In addition, the inner minimization over $t$ and $\alpha_1$ contributes an extra factor of $|T| \cdot k$ to the runtime.
For uniform jobs, this gives runtime $O(B \cdot k^2 \cdot n^7)$ and the space complexity $O( B \cdot k \cdot n^5)$.
and for unit jobs, the runtime is $O(B \cdot k^2 \cdot n^4)$ and the space complexity $O( B \cdot k \cdot n^3)$ . 
 
\end{proof}

This DP exemplifies that while one can adapt Baptiste's framework to include active time slots, it can be expensive without the additional structural properties guaranteed by the agreeable deadlines assumption.

\section{Conclusion and Future Work}\label{sec: conclusion}

We showed a fast dynamic programming algorithm with runtime $O(B \cdot k \cdot n)$ for minimizing the flow time of scheduling $n$ unit jobs with agreeable deadlines on a single processor that can process up to $B$ jobs at once given a budget of $k$ active batches. For the more general uniform jobs setting, our DP has runtime $O(B \cdot k \cdot n^5)$. Additionally, we can modify the DPs for the setting when the algorithm is allowed to only schedule $m \leq n$ jobs.
This is the first work that balances the wins between flow time and active time minimization, which intuitively work against each other.

For our techniques, the agreeable deadlines assumption lends to substantially better runtimes, as we are able to find a total ordering on the set of jobs that determines their scheduling order.
However, it would be interesting to see whether algorithms exist that can quickly solve the settings with arbitrary deadlines either optimally or approximately, as the modified Baptiste algorithm is prohibitively slow.
We discussed the challenge behind modifying the the Lazy Activation algorithm to obtain a solution with better flow time, but this could be a possible approach as well. If one considered arbitrary length jobs, then an LP based approach could also be interesting, but the integrality gap for LPs containing active time constraints is currently 2. So given our current understanding, one would have to be willing to lose some optimality in the number of active time slots in order to obtain flow time guarantees with this approach.

{\bf Acknowledgements:}
We are grateful to Dr. Jessica Chang for valuable early discussions on this problem and to Benjamin G. Schiffer for discussions on the problem and helpful comments on earlier versions of the paper.


\end{document}